\documentclass[a4paper,11pt]{article}
\usepackage{fullpage}
\usepackage{algorithm}
\usepackage{newfloat}
\usepackage{listings}
\usepackage[export]{adjustbox}
\usepackage{amsfonts}       %
\usepackage{amsmath}
\usepackage{amssymb}
\usepackage{amsthm}
\usepackage{array}
\usepackage{booktabs}       %
\usepackage{caption}
\usepackage{cite}
\usepackage{dblfloatfix}
\usepackage{enumitem}
\usepackage{xcolor}
\usepackage[pdfstartview=FitH,
pdfpagemode=UseNone,
colorlinks,
linkcolor={red!60!black},
citecolor={blue!80!black},
urlcolor={blue!90!black}]{hyperref}
\usepackage[nameinlink]{cleveref}
\usepackage{microtype}      %
\usepackage{multirow}
\usepackage{nicefrac}       %
\usepackage[caption=false,font=footnotesize]{subfig}
\usepackage{tikz}
\usepackage[textsize=tiny]{todonotes}
\usepackage{url}
\usepackage{xfrac}
\usepackage{calc}
\usepackage[noend]{algpseudocode}
\usepackage{siunitx}
\usepackage{pgfpages}
\usepackage{amsmath}
\usepackage{amsthm}
\usepackage{amssymb}
\usepackage{epsfig}
\usepackage{tikz}
\theoremstyle{definition}

\newtheorem{theorem}{Theorem}[section]
\newtheorem{lemma}[theorem]{Lemma}
\newtheorem{proposition}[theorem]{Proposition}
\newtheorem{definition}[theorem]{Definition}

\newcommand{\satisfying}[1]{\ensuremath{\mathsf{Sol}(#1)}}
\newcommand{\projectsatisfying}[2]{\ensuremath{\mathsf{Sol}(#1)_{\downarrow{#2}}}}
\newcommand{\range}[1]{\ensuremath{\mathsf{range}(#1)}}

\newcommand{\arijit}[1]{}

\newcommand{\sklmc}{\ensuremath{\mathsf{SkolemFC}}}
\newcommand{\base}{\ensuremath{\mathsf{Baseline}}}
\newcommand{\qc}{\ensuremath{\mathsf{QCounter}}}

\newcommand{\apxmc}{\ensuremath{\mathsf{ApproxMC}}}

\newcommand{\ganak}{\ensuremath{\mathsf{Ganak}}}

\newcommand{\unisamp}{\ensuremath{\mathsf{UniSamp}}}
\newcommand{\manthan}{\ensuremath{\mathsf{Manthan}}}
\newcommand{\ausampor}{\ensuremath{\mathsf{AlmostUniformSample}}}
\newcommand{\apxcountor}{\ensuremath{\mathsf{ApproxCount}}}
\newcommand{\excountor}{\ensuremath{\mathsf{Count}}}

\newcommand{\np}{\ensuremath{\mathsf{NP}}}
\newcommand{\sat}{\ensuremath{\mathsf{SAT}}}
\newcommand{\tot}{375}

\newcommand{\skolempsix}{\ensuremath{\mathsf{Skolem}(F,Y)}}
\newcommand{\countpsix}{\ensuremath{|\skolempsix|}}

\newcommand{\iter}{\ensuremath{{t}}}

\newcommand{\solproj}[2]{\ensuremath{\mathsf{Sol}(#1)_{\downarrow{#2}}}}

\newcommand{\est}{\ensuremath{\mathsf{Est}}}

\newcommand{\cam}[1]{{#1}}
\newcommand{\sol}{\ensuremath{\mathsf{Sol}}}
\newcommand{\aest}{\ensuremath{\mathsf{EEst}}}

\newcommand{\solcond}[1]{\ensuremath{\mathsf{Sol}(F \wedge (X=#1))}}
\newcommand{\solcondsig}{\solcond{\sigma}}

\newcommand{\epc}{\ensuremath{\varepsilon_c}}
\newcommand{\eps}{\ensuremath{\varepsilon_s}}

\newcommand{\oeps}{\ensuremath{(1 + \eps)}}

\newcommand{\la}{\ensuremath{\mathsf{LA}(\sigma)}}
\newcommand{\lc}{\ensuremath{\mathsf{LC}(\sigma)}}
\newcommand{\siglc}{\ensuremath{\displaystyle\sum_{\sigma \in S_2} \lc}}

\newcommand{\sigxlc}{\ensuremath{\displaystyle\sum_{\sigma \in 2^X} \lc}}
\newcommand{\sigslc}{\ensuremath{\displaystyle\sum_{\sigma \in S_s} \lc}}
\newcommand{\sigsla}{\ensuremath{\displaystyle\sum_{\sigma \in S_s} \la}}

\begin{document}
\title{\bf An Approximate Skolem Function Counter\thanks{  A preliminary version of this work appears at  Annual AAAI Conference on Artificial Intelligence, AAAI, 2024. The tool {\sklmc} is available at \url{https://github.com/meelgroup/skolemfc}.}}
\author{
\textbf{Arijit Shaw}\\
Chennai Mathematical Institute\\
IAI, TCG-CREST, Kolkata
\and
\textbf{Brendan Juba}\\
Washington University in St. Louis\\
\and
\textbf{Kuldeep S. Meel}\\
University of Toronto
}
\date{}
\maketitle

\begin{abstract}

One approach to probabilistic inference involves counting the number of models of a given Boolean formula. Here, we are interested in inferences involving higher-order objects, i.e., functions. We study the following task: Given a Boolean specification between a set of inputs and outputs, count the number of functions of inputs such that the specification is met. Such functions are called Skolem functions.

We are motivated by the recent development of scalable approaches to Boolean function synthesis. This stands in relation to our problem analogously to the relationship between Boolean satisfiability and the model counting problem. Yet, counting Skolem functions poses considerable new challenges. From the complexity-theoretic standpoint, counting Skolem functions is not only $\#P$-hard; it is quite unlikely to have an FPRAS (Fully Polynomial Randomized Approximation
Scheme) as the problem of synthesizing a Skolem function remains challenging, even given access to an NP oracle.

The primary contribution of this work is the first algorithm, {\sklmc}, that computes an estimate of the number of Skolem functions. {\sklmc} relies on technical connections between counting functions and propositional model counting: our algorithm makes a linear number of calls to an approximate model counter and computes an estimate of the number of Skolem functions with theoretical guarantees. Moreover, we show that Skolem function count can be approximated through a polynomial number of calls to a {\sat} oracle. Our prototype displays impressive scalability, handling benchmarks comparably to state-of-the-art Skolem function synthesis engines, even though counting all such functions ostensibly poses a greater challenge than synthesizing a single function.

\end{abstract}

\section{Introduction}

Probabilistic inference problems arise throughout AI and are tackled algorithmically by casting them as problems such as model counting~\cite{GSS21,CMV21}. In this work, we are interested in approaching inference questions for higher-order objects, specifically \emph{Skolem functions}: that is, we wish to compute the number of possible Skolem functions for a given specification $F(X, Y)$. \cam{Counting Skolem functions is the natural analog of \#SAT for Skolem functions, and has been recently explored in the context of counting level-2 solutions for valid Quantified Boolean Formulas (QBFs)~\cite{PMS23}.}

More precisely, recall that given two sets $X = \{x_1, \dots , x_n\}$ and $Y = \{y_1, \dots , y_m\}$ of variables and a Boolean formula $F(X, Y)$ over $X \cup Y$, the problem of Boolean functional synthesis is to compute a vector $\Psi = \langle\psi_1, \dots , \psi_m\rangle$ of Boolean functions $\psi_i$, often called Skolem functions, such that $\exists Y F(X, Y) \equiv F(X, \Psi(X))$. Informally, given a specification between inputs and outputs, the task is to synthesize a function vector $\Psi$ that maps each assignment of the inputs to an assignment of the outputs so that the combined assignment meets the specification (whenever such an assignment exists).
Skolem synthesis is a fundamental problem in formal methods and has been investigated by theoreticians and practitioners alike over the past few decades. The past few years have witnessed the development of techniques that showcase the promise of scalability in their ability to handle challenging specifications~\cite{JLH09,TV17,RTRS18,AACKRS19,GSRM21}.

The scalability of today's Skolem synthesis engines is reminiscent of the scalability of SAT solvers in the early 2000s. Motivated by the scalability of SAT solvers~\cite{FHIJ21}, researchers sought algorithmic frameworks for problems beyond satisfiability, such as MaxSAT~\cite{ABL13, LM21}, model counting~\cite{GSS21,CMV21}, sampling~\cite{CMV14}, and the like. The development of scalable techniques for these problems also helped usher in new applications, even though the initial investigation had not envisioned many of them.
In a similar vein, motivated in part by this development of scalable techniques for functional synthesis, we investigate the Skolem counting problem. We observe in Section~\ref{sec:application} that algorithms for such tasks also have potential applications in security and the engineering of specifications.
Being a natural problem, we will see that our study also naturally leads to deep technical connections between counting functions and counting propositional models and the development of new techniques, which is of independent interest.

Counting Skolem functions indeed raises new technical challenges. The existing techniques developed in the context of propositional model counting either construct (implicitly or explicitly) a representation of the space of all models~\cite{T06,SRMM19,DPV20} or at least enumerate a small number of models~\cite{CMV13,CMV16, SM19,YM23}, which in practice amounts to a few tens to hundreds of models of formulas constrained by random XORs. Such approaches are unlikely to work efficiency in the context of Skolem function counting, where even finding one Skolem function is hard.%

\subsection{Technical Contribution}

The primary contribution of this work is the development of a novel algorithmic framework, called {\sklmc}, that approximates the Skolem function count with a theoretical guarantee, using only linearly many calls to an approximate model counter and almost-uniform sampler.
First, we observe that Skolem function counting can be reduced to an exponential number of model counting calls, serving as a baseline Skolem function counter.
The core technical idea of {\sklmc} is to reduce the problem of approximate Skolem function counting to only linearly many (in $m=|Y|$) calls to propositional model counters. Of particular note is the observation that {\sklmc} can provide an approximation to the number of Skolem functions without enumerating even one Skolem function.
As approximate model counting and almost-uniform sampling can be done by logarithmically many calls to {\sat} oracle, we show that Skolem function counting can also be reduced to polynomially many calls to a {\sat} oracle.

To measure the impact of the algorithm, we implement {\sklmc} and demonstrate its potential over a set of benchmarks arising from prior studies in the context of Skolem function synthesis. Out of 609  instances, {\sklmc} could solve {\tot} instances, while a baseline solver could solve only eight instances. For context, the state-of-the-art Skolem function synthesis tool {\manthan$\mathsf{2}$}~\cite{GSRM21} effectively tackled 509 instances from these benchmarks, while its precursor, {\manthan}~\cite{GRM20}, managed only 356 instances  with a timeout of 7200 seconds.
\subsection{Applications}\label{sec:application}
This problem arises in several potential application areas.

\paragraph{Specification engineering.}
The first and primary motivation stems from the observation that specification synthesis~\cite{ADG16,PFMD21} (i.e., the process of constructing $F(X, Y)$) and function synthesis form part of the iterative process wherein one iteratively modifies specifications based on the functions that are constructed by the underlying engine. In this context, one helpful measure is to determine the number of possible semantically different functions that satisfy the specification, as often a large number of possible Skolem functions indicates the vagueness of specifications and highlights the need for strengthening the specification. Note that the use of the count is qualitative here, and hence an approximate order of magnitude (log count) suffices.

\paragraph{Diversity at the specification level.}
In system security and reliability, a classic technique is to generate and use a diverse variety of functionally equivalent implementations of components~\cite{baudry2015multiple}. Although classically, this is achieved by transformations of the code that preserve the function computed, we may also be interested in producing a variety of functions that satisfy a common specification. Unlike transformations on the code, it is not immediately clear whether a specification even admits a diverse collection of functions -- indeed, the function may be uniquely defined. Thus, counting the number of such functions is necessary to assess the potential value of taking this approach, and again a rough order of magnitude estimate suffices. Approximate counting of the functions may also be a useful primitive for realizing such an approach.

\paragraph{Uninterpreted functions in SMT.}
A major challenge in the design of counting techniques for SMT~\cite{CDM15, CMMV16} lies in handling uninterpreted functions~\cite{KS16}. Since Skolem functions capture a restricted but large enough class of uninterpreted functions (namely, the case where a given uninterpreted function is allowed to depend on all $X$ variables), progress in Skolem function counting is needed if we hope to make progress on the general problem of counting of uninterpreted functions in SMT.

\paragraph{Evaluation of a random Skolem function.}
Although synthesis of Skolem functions remains challenging in general, we note that approximate counting enables a kind of incremental evaluation by using the standard techniques for reducing sampling to counting. More concretely, given a query input, we can estimate the number of functions that produce each output: this is trivial if the range is small (e.g., Boolean), and otherwise, we can introduce random XOR constraints to incrementally specify the output. Once an output is specified for the query point, we may retain these constraints when estimating the number of consistent functions for subsequent queries, thereby obtaining an approximately uniform function conditioned on the answers to the previous queries. %

\subsection{Organization}
The rest of the paper is organized as follows:  we discuss related work in Section~\ref{sec:related} and present notation and preliminaries in  Section~\ref{sec:background}. We then present the primary technical contribution of our work in Section~\ref{sec:algo}. We present the empirical analysis of the prototype implementation of {\sklmc} in Section~\ref{sec:results}. We finally conclude in Section~\ref{sec:concl}.
\section{Related Work}\label{sec:related}
Despite the lack of prior studies focused on the specific problem of counting Skolem functions, significant progress has been made in synthesizing these functions. Numerous lines of research have emerged in the field of Skolem function synthesis.
The first, incremental determinization, iteratively pinpoints variables with distinctive Skolem functions, making decisions on any remaining variables by adding provisional clauses that render them deterministic~\cite{RS16,RTRS18,R19}. The second line of research involves obtaining Skolem functions by eliminating quantifiers using functional composition and reducing the size of composite functions through the application of Craig interpolation~\cite{JLH09,J09}. The third, CEGAR-style approaches, commence with an initial set of approximate Skolem functions, and proceed to a phase of counter-example guided refinement to improve upon these candidate functions~\cite{JSCTA15,ACJS17,ACGKS18}.
Work on the representation of specification $F(X,Y)$ has led to efficient synthesis using ROBDD representation and functional composition~\cite{BJ11}, with extensions to factored specifications~\cite{TV17,CFTV18}. Notable advancements include the new negation normal form, SynNNF, amenable to functional synthesis~\cite{AACKRS19}.
Finally, a data-driven method has arisen~\cite{GRM20,GSRM21}, relying on constrained sampling to generate satisfying assignments for a formula $F$. %

The specifications for the functions are expressed in terms of quantified Boolean formulas (QBFs). The problem of counting the number of solutions in true  QBF formulas has been studied in different definitions.  \cite{BELM12} defined the problem of AllQBF - finding all assignments of the free variables of a given QBF such that the formula evaluates to true. CountingQBF~\cite{SMKS22} poses a similar query - counting the number of assignments of the free variables of a given QBF such that the formula evaluates to true. However, their relevance to counting functions isn't clear. \cam{In more recent work, \cite{PMS23} investigated the problem of \textit{level-2 solution counting} for QBF formulas and developed an enumeration-based approach for counting the number of level-2 solutions for $\forall \exists$-QBFs. In this work, the solutions are expressed as tree models, and for true formulas, the number of \textit{different} tree models is the same as the number of Skolem functions.}

A related problem in the descriptive complexity of functions definable by counting the Skolem functions for \emph{fixed} formulas have been shown to characterize $\#AC_0$~\cite{HV19}. By contrast, we are interested in the problem where the \emph{formula} is the input.
Our algorithm also bears similarity to the FPRAS proposed for the descriptive complexity class $\#\Sigma_1$ \cite{DHKV21}, which is obtained by an FPRAS for counting the number of \emph{functions} satisfying a DNF over atomic formulas specifying that the functions must/must not take specific values at specific points. Nevertheless, our problem is fundamentally different in that it is easy to find functions satisfying such DNFs, whereas synthesis of Skolem functions is unlikely to be possible in polynomial time.

\section{Notation and Preliminaries} \label{sec:background}

We use lowercase letters (with subscripts) to denote propositional variables and
uppercase letters to denote a subset of variables. The formula $\exists Y F (X, Y )$ is existentially quantified in $Y$ , where $X = \{x_1, \dots , x_n\}$ and $Y = \{y_1, \dots , y_m\}$. By $n$ and $m$ we denote the number of $X$ and $Y$ variables in the formula. Therefore, $n = |X|, m = |Y|$.
For simplicity, we write a formula $F(X,Y)$ as $F$ if the $X,Y$ is clear from the context.
A \emph{model} is an assignment (true or false) to all the variables in $F$, such that $F$ evaluates to true.
Let $\solproj{F}{S}$ denote the set of models of formula $F$ projected on $S \subseteq X \cup Y$. If $S = X \cup Y$, we write the set as $\sol(F)$.
Let $\sigma$ be a partial assignment for the variables $X$ of $F$. Then $\solcond{\sigma}$ denotes the models of $F$ where $X = \sigma$.

\paragraph{{\np} Oracles and {\sat} oracles.}
Given a Boolean formula $F$, an {\np} oracle determines the satisfiability of the formula. A SAT solver is a practical tool solving the problem of satisfiability. Following definition of Delannoy and Meel~\cite{DM22}, a {\sat} oracle takes in a formula $F$ and returns a satisfying assignment $\sigma$ if $F$ is satisfiable and $\bot$ otherwise. The {\sat} oracle model captures the behavior of the modern SAT solvers. %

\paragraph{Propositional Model Counting.}
Given a formula $F$ and a projection set $S$, the problem of model counting is to compute $|\solproj{F}{S}|$.
An  \emph{approximate model counter} takes in a formula $F$, projection set $S$, tolerance parameter $\varepsilon$, and confidence parameter $\delta$, and returns $c$ such that $\Pr\left[\frac{|\solproj{F}{S}|}{1+\varepsilon} \leq c \leq (1+\varepsilon) |\solproj{F}{S}| \right] \geq 1-\delta$. It is known that $\log(n)$ calls to a {\sat} oracle are necessary~\cite{CCKM23}  and sufficient~\cite{CMV16} to achieve $(\varepsilon, \delta)$ guarantees for approximately counting the models of a formula with $n$ variables.

\paragraph{Propositional Sampling.}
Given a Boolean formula $F$ and a projection set $S$, a \emph{sampler} is a probabilistic algorithm that generates a random element in \solproj{F}{S}.
An \emph{almost uniform sampler} $G$ takes a tolerance parameter $\varepsilon$ along with $F$ and $S$, and  guarantees $ \forall y \in \solproj{F}{S}, \frac{1}{ (1+\varepsilon) |\solproj{F}{S}|} \leq \Pr [G(F,S,\varepsilon) = y] \leq \frac{(1+\varepsilon)} {|\solproj{F}{S}|}$. Delannoy and Meel~\cite{DM22} showed, $\log(n)$ many calls to a {\sat} oracle suffices to generate almost uniform samples from a formula with $n$ variables.

\paragraph{Skolem Functions.}

Given a Boolean specification $F(X, Y)$ between set of inputs $X = \{x_1, \dots , x_n\}$ and vector of outputs $Y = \langle y_1, \dots , y_m\rangle$, a function vector $\Psi(X) = \langle \psi_1(X), \psi_2(X), \dots $ $, \psi_m(X) \rangle$ is a Skolem function vector if $y_i \leftrightarrow \psi_i(X)$ and $\exists Y F(X, Y) \equiv F(X, \Psi)$. We refer to $\Psi$ as the Skolem function vector and $\Psi_i$ as the Skolem function for $y_i$. We'll use the notation  $\skolempsix$  to denote the set of possible $\Psi(X)$ satisfying the condition   $\exists YF(X, Y) \equiv F(X, \Psi(X))$. Two Skolem function vectors $\Psi_1$ and $\Psi_2$ are different, if there exists an assignment $\sigma \in {\solproj{F}{X}}$ for which $\Psi_1(\sigma) \neq \Psi_2(\sigma)$.

For a specification $\exists Y F(X,Y)$, the number of Skolem functions itself can be as large as $2^{n\cdot 2^m}$,
and the values of $n$ and $m$ are quite large in many practical cases. Beyond being a theoretical possibility, the count of Skolem functions is often quite big, and such values are sometimes difficult to manipulate and store as 64-bit values. Therefore, we are interested in the logarithm of the counts, and define the problem of approximate Skolem function counting as following:

\paragraph{Problem Statement.} Given a Boolean specification $F(X,Y)$, tolerance parameter $\varepsilon$, confidence parameter $\delta$, let $\ell = \log(\countpsix)$, the task of approximate Skolem function counting is to give an estimate $\est$, such that
    $\Pr\left[{(1-\varepsilon){\ell}} \leq \est \leq (1+\varepsilon) \ell \right] \geq 1-\delta$.

In practical scenarios, the input \emph{specification} is often given as a quantified Boolean formula (QBF). The output of the synthesis problem is a function, which is expressed as a Boolean circuit. In our setting, even if two functions have different circuits, if they have identical input-output behavior, we consider them to be the same function. For example, let $f_1(x) = x$ and $f_2 = \neg(\neg x)$. We'll consider $f_1$ and $f_2$ to be the same function.

\paragraph{Illustrative Example.}
Let's examine a formula defined on three sets of variables $X,Y_0,Y_1$, where each set contains five Boolean variables, interpreted as five-bit integers: $\exists Y_0Y_1 F(X, Y_0Y_1)$, where   $F$ represents the constraint for factorization,  $X=Y_0\times Y_1, Y_0 \leq Y_1, Y_0 \neq 1$. The number of Skolem functions of $F$ gives the number of distinct ways to implement a factorization function for 5-bit input numbers.
There exist multiple $X$'s for which there are multiple factorizations: A Skolem function $S_1$ may factorize $12$ as $4 \times 3$ and a function $S_2$ may factorize $12$ as $2 \times 6$.

\paragraph{Stopping Rule Algorithm.}
Let $Z_1,Z_2, \dots$ denote independently and identically distributed (i.i.d.) random variables taking values in the interval $[0,1]$ and with mean $\mu$. Intuitively, $Z_t$ is the outcome of experiment $t$. Then the Stopping Rule algorithm (\Cref{alg:stop}) approximates $\mu$ as stated by \Cref{th:stop} \cite{DKLR95,DL97}.

\begin{algorithm}[tb]
    \caption{Stopping Rule $(\varepsilon, \delta)$}\label{alg:stop}
    \begin{algorithmic}[1]
        \State  $t \gets 0, x \gets 0, s \gets 4 \ln (2/\delta)(1+\varepsilon)/\varepsilon^2$

        \While{{$x < s$}}
        \State  $t \gets t +1$
        \State Generate Random Variable $Z_t$
        \State $x \gets x + Z_t$
        \EndWhile
        \State $\est \gets s/t$
        \State \Return \est
    \end{algorithmic}
\end{algorithm}

\begin{theorem}[Stopping Rule Theorem]\label{th:stop} For all $0 < \varepsilon \leq 2, \delta > 0$, if the Stopping Rule algorithm returns {\est}, then
    $\Pr [ \mu (1-\varepsilon) \leq \est \leq \mu (1+ \varepsilon) ] > (1-\delta) $.
\end{theorem}

\paragraph{FPRAS.} A Fully Polynomial Randomized Approximation Scheme (FPRAS) is a randomized algorithm that, for any fixed $\varepsilon > 0$ and any fixed probability $\delta > 0$, produces an answer that is within a factor of $(1+\varepsilon)$ of the correct answer, and does so with probability at least $(1-\delta)$, in polynomial time with respect to the size of the input, $1/\varepsilon$, and $\log(1/\delta)$.

\section{Algorithm} \label{sec:algo}
In this section, we introduce the primary contribution of our paper: the {\sklmc} algorithm. The algorithm takes in a formula $F(X,Y)$ and returns an estimate for $\log(|\skolempsix|)$. We first outline the key technical ideas that inform the design of  {\sklmc} and then present the pseudocode for implementing this algorithm.

\subsection{Core Ideas} \label{subsec:coreidea}

Since finding even a single Skolem function is computationally expensive, our approach is to estimate the count of Skolem functions without enumerating even a small number of Skolem functions. The key idea is to observe that the number of Skolem functions can be expressed as a product of the model counts of formulas.
A Skolem function $\Psi \in \skolempsix$ is a function from  $2^{X}$ to $2^{Y}$. A useful quantity in the context of counting Skolem functions is to define, for every assignment $\sigma \in 2^{X}$, the set of elements in $2^{Y}$ that $\Psi(X)$ can belong to. We refer to this quantity as $\range{\sigma}$ and formally define it as follows:

\begin{definition}
    \begin{align*}
        \range{\sigma} = \begin{cases}
            \projectsatisfying{F\wedge(X=\sigma)}{Y}  & |\solcondsig| > 0  \\
            1 & \text{otherwise}
        \end{cases}
    \end{align*}
    \end{definition}

\begin{lemma}\label{lemma:sk}
    $\countpsix = \prod\limits_{\sigma \in 2^{X}} |\range{\sigma}|$
\end{lemma}

\begin{proof}
    First of all, we observe that
    \begin{align*}
        \forall \sigma \in 2^{X}, \forall \pi \in \range{\sigma}, \exists \Psi \text{ s.t. } \Psi(\sigma) = \pi
    \end{align*}
    which is easy to see for all $\sigma \in 2^{X}$ for which there exists $\pi \in 2^{Y}$ such that $F(\sigma,\pi) = 1$. As for $\sigma \in 2^{X}$ for which there is no $\pi$ such that $F(\sigma,\pi) = 1$,
Skolem functions that differ solely on inputs $\sigma \notin \solproj{F}{X}$ are regarded as identical. Consequently, such inputs have no impact on the count of distinct Skolem functions, resulting in $\range{\sigma} = 1$ for these cases. Recall, {\skolempsix} is the set of all functions from $2^X$ to $2^Y$. It follows that
    $\countpsix = \prod\limits_{\sigma \in 2^{X}} |\range{\sigma}|$.
\end{proof}

~\Cref{lemma:sk} allows us to develop a connection between Skolem function counting and propositional model counting.
As stated in the problem statement, we focus on estimating $\log \countpsix$.
To formalize our approach, we need to introduce the following notation:

\begin{proposition}\label{prop:logsk}
    \begin{align*}
        &\log \countpsix = \sum\limits_{\sigma \in S_2} \log |\satisfying{F \wedge (X = \sigma)}|\\
        &\qquad \text{where, }  S_2 := \{\sigma \in 2^{X} \mid  |\solcondsig| \geq 2 \}
    \end{align*}

\end{proposition}
\begin{proof}
    From ~\Cref{lemma:sk}, we have 	$\countpsix = \prod\limits_{\sigma \in 2^{X}} |\range{\sigma}|$. Taking logs on both sides, partitioning $2^{X}$ into $ S_2$ and $2^X \backslash S_2$, and observing that $\log |\range{\sigma}| = 0 $ for $\sigma \notin S_2$, we get the desired result.
\end{proof}

\begin{algorithm}[!tb]
    \caption{\sklmc $(F(X,Y),\varepsilon, \delta)$}\label{alg:skolemmc}
    \begin{algorithmic}[1]
        \State $\varepsilon_f \gets 0.6\varepsilon, \delta_f \gets 0.4\delta, s \gets 4 \ln (2/\delta_f)(1+\varepsilon_f)/{\varepsilon_f}^2$

        \State $\varepsilon_s \gets 0.2\varepsilon, \delta_c \gets 0.4\delta/ms, \varepsilon_c \gets 4\sqrt{2} - 1$

        \State $\varepsilon_g \gets 0.1\varepsilon, \delta_g \gets 0.1\delta$

        \State $G(X,Y,Y') := F(X,Y) \wedge F(X,Y') \wedge (Y \neq Y')$ \label{algline:eqn}

        \While{$x < s$} \label{algline:loops}

        \State  $\sigma \leftarrow  {\ausampor}(G,X,\eps)$  \label{algline:sample}

        \State $c\leftarrow \log( {\apxcountor}(F \wedge (X = \sigma),\epc,\delta_c))/m$ \label{algline:count}

        \State $x \leftarrow x + c $ \label{algline:sum}

        \State $t \leftarrow t + 1 $ \label{algline:inct}

        \EndWhile \label{algline:loopn}

        \State $g \gets {\apxcountor}(G,X,\varepsilon_g, \delta_g)$ \label{algline:getgcnt}
        \State$\est \gets  s / {\iter} \times m \times g $ \label{algline:gets2est}

        \If{$(g\log(1+\epc) > 0.1\est)$} \Return  $\bot$ \EndIf \label{algline:countcheck}

        \State\Return $\est$

    \end{algorithmic}
\end{algorithm}

\subsection{Algorithm Description}

The pseudocode for {\sklmc} is delineated in ~\Cref{alg:skolemmc}. It accepts a formula $\exists Y F(X,Y)$, a tolerance level $\varepsilon$, and a confidence parameter $\delta$. The algorithm {\sklmc} then provides an approximation of $\log \countpsix$ following ~\Cref{prop:logsk}.
To begin, {\sklmc} almost-uniformly samples $\sigma$ from $S_2$ at random in line~\ref{algline:sample}, utilizing an almost-uniform sampler. Subsequently, {\sklmc} approximates $|\satisfying{F \wedge (X = \sigma)}|$ through an approximate model counter at line~\ref{algline:count}.
The estimate $\est$ is computed by taking the product of the mean of $c$'s and $|S_2|$. In order to sample $\sigma \in S_2$, {\sklmc} constructs the formula $G$ whose solutions, when projected to $X$ represent all the assignments $\sigma \in S_2$ (line~\ref{algline:eqn}). Finally, {\sklmc} returns the estimate {\est} as logarithm of the Skolem function count.

 The main loop of {\sklmc} (from lines \ref{algline:loops} to \ref{algline:loopn}) is based on the Stopping Rule algorithm presented in \Cref{sec:related}. The Stopping Rule algorithm is utilized to approximate the mean of a collection of i.i.d.\ random variables that fall within the range $[0,1]$. The method repeatedly adds the outcomes of these variables until the cumulative value reaches a set threshold $s$. This threshold value is influenced by the input parameters $\varepsilon$ and $\delta$. The result yielded by the algorithm is represented as $s/t$, where $t$ denotes the number of random variables aggregated to achieve the threshold $s$. In the context of {\sklmc}, this random variable is defined as $\log( {\apxcountor}(F \land (X = \sigma),\epc,\delta_c))/m$. Line~\ref{algline:countcheck} asserts that the error introduced by the approximate model counting oracle is within some specific bound.

\paragraph{Oracle Access.}
We assume access to approximate model counters and almost-uniform samplers as oracles. The notation {\apxcountor}$(F,P,\varepsilon,\delta)$ represents an invocation of the approximate model counting oracle on Boolean formula $F$ with a projection set $P$, tolerance parameter $\varepsilon$, and confidence parameter $\delta$.
${\ausampor}(F,S,\varepsilon)$ denotes an invocation of the almost uniform sampler on a formula $F$, with projection set $S$ and tolerance parameter $\varepsilon$.
The particular choice of values of $\varepsilon_s$, $\varepsilon_c$, $\delta_c$, $\varepsilon_g$, $\delta_g$ used in the counting and sampling oracle aids the theoretical guarantees.

\subsection{Illustrative Example}
We will now examine the specification of factorization as outlined in \Cref{sec:background}, and investigate how {\sklmc} estimates the count of Skolem functions meeting that specification.

\begin{enumerate}
    \item In line~\ref{algline:eqn}, {\sklmc} constructs $G$ such that $|\solproj{G}{X}| = 7$, $\solproj{G}{X} = \{ 12, 16,18,20,24,28,30\}$.
    \item In line~\ref{algline:sample}, it samples $\sigma$ from $\solproj{G}{X}$. Let's consider $\sigma = 30$. Then $\solcondsig_{\downarrow Y_0,Y_1} = \{(2,15),(3,10),(5,6)\}.$ Therefore, $c = \log(3)$ in line~\ref{algline:count}.
    \item Suppose in the next iteration it samples $\sigma = 16.$  Then $\solcondsig = \{(2,8),(4,4)\}.$ Therefore, $c = \log(2)$ in line~\ref{algline:count}.
    \item Now suppose that the termination condition of line
    \ref{algline:loops} is reached. At this point, the estimate {\est} returned from line~\ref{algline:gets2est} will be $\approx \frac{(\log(3) + \log(2))}{2} \times 7 \approx 6.$
    \item Finally, {\sklmc} will return the value {\est = 6.}

\end{enumerate}

Note that the approach is in stark contrast to the state-of-the-art counting techniques in the context of propositional models, which either construct a compact representation of the entire solution space or rely on the enumeration of a small number of models.

\subsection{Analysis of {\sklmc}} \label{sec:analysis}
Let $F(X,Y)$ be a propositional CNF formula over variables $X$ and $Y$. In the section we'll show that {\sklmc} works as an approximate counter for the number of Skolem functions.
We create a formula $G(X,Y,Y') = F(X,Y) \wedge F(X,Y') \wedge (Y \neq Y')$ from $F(X,Y)$, where $Y'$ is a fresh set of variables and $m = |Y'|$. We show that, if we pick a solution from $G(X,Y,Y') = F(X,Y) \wedge F(X,Y') \wedge (Y \neq Y')$, then the assignment to $X$ in that solution to $G$  will have at least two solutions in $F(X,Y)$.

\begin{lemma}\label{lemma:g}
\begin{equation*}
\small
\projectsatisfying{G}{X} = \left\{ \sigma \mid \sigma \in 2^{X} \wedge |\solcondsig| \geq 2 \right\}
\end{equation*}
\end{lemma}

\begin{proof}
    We can write the statement alternatively as,
\begin{equation*}
\small
    \sigma \in \mathsf{Sol}(G,X) \iff |\solcondsig| \geq 2
\end{equation*}
    $(\implies)$ For every element $\sigma \in \mathsf{Sol}(G)$, we write $\sigma$ as  $\langle \sigma_{\downarrow X},\sigma_{\downarrow Y}, \sigma_{\downarrow Y'} \rangle$. Now according to the definition of $G$, both $\langle \sigma_{\downarrow X}, \sigma_{\downarrow Y} \rangle$ and $\langle \sigma_{\downarrow X}, \sigma_{\downarrow Y'} \rangle$ satisfy $F$. Moreover, $\sigma_{\downarrow Y}$ and $\sigma_{\downarrow Y'}$ are not equal. Therefore, $ |\solcondsig| \geq 2$.

    $(\impliedby)$  If $ |\solcondsig| \geq 2$, then $F(X,Y)$ has  solutions of the form $\langle \sigma, \gamma_1 \rangle$ and $\langle \sigma, \gamma_2 \rangle$, where $\gamma_1 \neq \gamma_2$. Now  $\langle \sigma \gamma_1 \gamma_2 \rangle$ satisfies $G$.
\end{proof}

\begin{theorem} \label{th:main}
{\sklmc} takes in input $F(X,Y)$, $\varepsilon > 0$, and $\delta \in (0,1]$, and returns $\est$ such that
\begin{align*}
	&\Pr\left[{(1-\varepsilon){\ell}} \leq \est \right. \left. \leq (1+\varepsilon) \ell \right] \geq 1-\delta
\end{align*}
 	 where, $\ell = \log(\countpsix)$. Furthermore, it makes $\tilde{O}\left(\frac{m}{\varepsilon^2} \ln \frac{2}{\delta}  \right)$ many calls to a {\sat} oracle, where $\tilde{O}$ hides poly-log factors in parameters $m, n, \varepsilon, \delta$.

\end{theorem}

\begin{proof}
    We prove the correctness and complexity of the algorithm in two different parts.
    \subsection*{Proof for Correctness}
    There are four sources of error in the approximation of the count. In the following table we list those. We'll use the following shorthand notation {\lc}  for $\log( |\sol(F \wedge X = \sigma)|)$.

\begin{table}[!htb]
    \centering
    \renewcommand{\arraystretch}{1.8}
    \begin{tabular}{p{5.6cm}|c}
        \toprule
        Reason of Error & Estimated Amount \\ \midrule
        Approximation by {\sklmc} $(E_F)$ &   $\varepsilon_f \siglc $    \\
        Non-uniformity in Sampling  $(E_S)$ & $\varepsilon_s \siglc $ \\
        Approximation by {\apxcountor}  in line~\ref{algline:count}$(E_M)$ &  $|S_2|\log(1+\varepsilon_c)$ \\
        Approximation by {\apxcountor} in line~\ref{algline:getgcnt}  $(E_G)$ &  $\varepsilon_g \siglc + \varepsilon_g E_M$ \\ \bottomrule
    \end{tabular}
\end{table}

In the following part, we will establish the significance of the bounds on individual errors, culminating in a demonstration of how these errors combine to yield the error bound for {\sklmc}.

\subsubsection*{(1) Bounding the error from \sklmc}

For a given formula $\exists Y F(X,Y)$, we partition $2^X$, the assignment space of $X$, into three subsets, $S_0, S_1, S_2$ as discussed in~\Cref{subsec:coreidea}.
In line~\ref{algline:loops} of \Cref{alg:skolemmc}, we pick an assignment to $X$ almost uniformly at random from $\solproj{G}{X}$. Let $\sigma_{i}$ be the random assignment picked at $i^{th}$ iteration of the for loop. For each of the random event of picking an assignment $\sigma_{i}$, we define a random variable $W_i$ such that, $W_i$ becomes a function of $\sigma_{i}$ in the following way:\\
    $$ W_i =  \displaystyle\frac{\log(|\solcond{\sigma_{i}}|)}{m}$$
    Therefore, we have $\iter$ random variables $W_1,W_2, \dots W_{\iter}$  on the sample space of $S_2$.
    Now as $\sigma_{i} $ is picked from $\solproj{G}{X}$ at line \ref{algline:eqn} of \sklmc, from \Cref{lemma:g} we know, $ 2 \leq |\solcond{\sigma_{i}}|$.
    Therefore the following holds,
    \begin{align*}
        \forall i, 2 \leq |\solcond{\sigma_{i}}| \leq 2^{m} \\
    \end{align*}
    Hence, from the definition of $W_i$:
    \begin{align*}
        \forall i, \frac{1}{m} \leq W_i \leq 1
    \end{align*}
Now assume that the model counting oracle of line~\ref{algline:count} returns an exact model count instead of an approximate model count; and the value retured at line~\ref{algline:gets2est} is {\aest} instead of {\est}.  Now, due to approximation by {\sklmc}, according to \Cref{th:stop} (stopping rule theorem),
\begin{equation}\label{eq:dklr}
    \Pr \left[ \left| \frac{\aest}{m|S_2|} - \mu \right| \leq \varepsilon_f \mu \right] \geq 1 - \delta_f
\end{equation}

\subsubsection*{(2) Bounding the error from sampling oracle}

Let $p_j$ be the probability an assignment $\sigma_j \in S_2$ appears. Then,
$$\mu = E[W_i] = \sum_{\sigma_{j} \in S_2} p_j \frac{\mathsf{LC}(\sigma_j)}{ m}$$
\\Now from the guarantee provided by the almost  oracle, we know
$$\displaystyle \frac{1}{\oeps |S_2|} \leq p_j \leq \frac{\oeps}{|S_2|}$$
Therefore,
$$\displaystyle \frac{\siglc}{ m|S_2| \oeps } \leq \mu  \leq \frac{\oeps \siglc}{m |S_2|}$$
\begin{equation} \label{eq:limmu}
    \left|m|S_2|\mu  - \siglc \right| \leq E_S
\end{equation}

\subsubsection*{(3) Bounding the error from counting oracle in line~\ref{algline:count}}

We'll use the following shorthand notation {\la}  for $\log( \apxcountor(F \wedge X = \sigma, \varepsilon_c, \delta_c))$.
From properties of {\apxcountor} oracle, for arbitrary $\sigma$,
$$\Pr [|\lc -\la| \leq \log(1+\varepsilon_c)] \geq 1 - \delta_c$$
\\Let $S_s \subseteq S_2$ be the set of $t$ samples picked up at line~\ref{algline:sample}. ($t$ is the number of iterations taken by the algorithm, as in line~\ref{algline:inct}):
\begin{equation*}
    \Pr \left[ \left| \sigsla - \sigslc \right| \leq  t\log(1+\varepsilon_c) \right] \geq 1 - t\delta_c
\end{equation*}
\\Now $\aest = \frac{|S_2|}{t}\sigslc$ and let $\est' = \frac{|S_2|}{t}\sigsla$. Therefore,
\begin{equation}\label{eq:limcnt}
    \Pr \left[ \left| \aest- \est' \right| \leq E_M \right] \geq 1 - t\delta_c
\end{equation}

\subsubsection*{(4) Bounding the error from counting oracle in line~\ref{algline:getgcnt}}

From the property of {\apxcountor} oracle running with parameters $(\varepsilon_g, \delta_g)$, we get
\begin{align*}
    \Pr\left[\frac{|S_2|}{1+\varepsilon_g} \leq g \leq (1+\varepsilon_g) |S_2| \right] \geq 1-\delta_g
\end{align*}
Now, $\est =  \displaystyle\frac{g}{t} \displaystyle \sum_{\sigma \in S_{s}} LA(\sigma)$, which gives us:
    \begin{align*}
\Pr \left[ \frac{\est'}{1 + \varepsilon_{g}} \leq \est \leq (1 + \varepsilon_{g}) \est' \right] \geq 1 - \delta_{g}
\end{align*}
\begin{equation}\label{eqn:gcounterror}
    \Pr \left[ |\est - \est'| \leq (1 + \varepsilon_{g}) \est' \right] \geq 1 - \delta_{g}
\end{equation}
Now, combining \Cref{eqn:gcounterror} with \Cref{eq:limcnt}, we get
\begin{equation}\label{eqn:appmcerrcomb}
    \Pr \left[ |\est - \aest| \leq E_M + E_G \right] \geq 1 - \delta_{g} - t\delta_c
\end{equation}

\subsubsection*{(5) Summing-up all the errors}

Combining \Cref{eq:dklr} and (\ref{eq:limmu})
     \begin{equation}\label{eq:three}
     \Pr \left[ \left| \aest - \siglc \right| \leq E_S + E_F \right] \geq 1 - \delta_f
 \end{equation}
\\Combining \Cref{eqn:appmcerrcomb} and (\ref{eq:three}), putting $\delta = \delta_f + \delta_g + t\delta_c$
    \begin{equation}\label{eq:five}
        \Pr \left[ \left| \est - \siglc \right| \leq E_S + E_M + E_F + E_G\right] \geq 1 - \delta
    \end{equation}
According to \Cref{alg:skolemmc},  $\varepsilon_s = 0.2\varepsilon$  and $\varepsilon_f = 0.6\varepsilon$.
Now, line ~\ref{algline:countcheck} asserts that $E_M \leq 0.1\varepsilon \sigxlc$. That makes $E_S + E_M + E_F + E_G \leq \varepsilon \siglc $. %
Therefore,
 $$\Pr \left[ \left| \est - \sigxlc \right| \leq \varepsilon \sigxlc \right] \geq 1 - \delta$$
 Or,
$$\Pr\left[{(1-\varepsilon){\ell}} \leq \est \leq (1+\varepsilon) \ell \right] \geq 1-\delta$$
where $\ell = \log(\countpsix)$.

\subsection*{Proof for Complexity}

    The formula $G$ in line~\ref{algline:eqn} of {\sklmc} can be constructed with a time complexity that is a polynomial in $|F|$. As the minimum value for $c$ in line~\ref{algline:count} is 1, the for loop, located between lines \ref{algline:loops} and \ref{algline:loopn}, is executed a maximum of $\iter = \left(  4m \ln (2/\delta_f)(1+\varepsilon_f)/\varepsilon_f^2\right)$ times.   During each iteration, the algorithm makes one call to an approximate counting oracle and an almost uniform sampling oracle, resulting in a total of $\iter$ calls to these oracles. In line~\ref{algline:getgcnt}, we invoke approximate model counting oracle once more. Following ~\cite{S83} and ~\cite{DM22}, $\log(m+n)$ many calls to a SAT oracle suffices to give both approximate counting and almost uniform sampling. Therefore, {\sklmc} returns within %
    $\tilde{O}\left(\frac{m}{\varepsilon^2} \ln \frac{2}{\delta}  \right)$ many calls to a {\sat} oracle, where $\tilde{O}$ hides poly-log factors in parameters $m, n, \varepsilon, \delta$.
    \end{proof}

\section{Experiments} \label{sec:results}

We conducted a thorough evaluation of the performance and accuracy of results of the {\sklmc} algorithm by implementing a functional prototype in C++.
The following experimental setup was used to evaluate the performance and quality of results of the {\sklmc} algorithm\footnote{All benchmarks and experimental data are available at \url{ https://doi.org/10.5281/zenodo.10404174}}.

\paragraph{Baseline.}

\cam{To assess the effectiveness of {\sklmc}, we conducted comparisons with the tool {\qc}~\cite{PMS23}, which is designed for level-2 solution counting. However, given that {\qc} is restricted to true QBFs, we aim to develop an additional baseline counter to broaden our evaluation scope.}

A possible approach to count Skolem functions, following \Cref{lemma:sk}, is given in \Cref{alg:baseline}. The {\excountor}$(F)$ oracle denotes an invocation of exact model counter. We implemented that to compare with {\sklmc}. In the implementation, we relied on the latest version of {\ganak}~\cite{SRMM19} to get the necessary exact model counts. We use a modified version of the SAT solver $\mathsf{CryptoMiniSat}$~\cite{SNC09} as $\mathsf{AllSAT}$ solver to find all solutions of a given formula, projected on $X$ variables. We call this implementation {\base} in the following part of the paper.

\begin{algorithm}[!h]
    \caption{\base $(F(X,Y))$}\label{alg:baseline}
    \begin{algorithmic}[1]

        \State $\est \gets 0$

        \State $G(X,Y,Y') := F(X,Y) \wedge F(X,Y') \wedge (Y \neq Y')$

        \State $\mathsf{SolG} \gets \mathsf{AllSAT}(G,X)$

        \For{\textbf{each} $\sigma \in \mathsf{SolG}$}

        \State $c\leftarrow \log( {\excountor}(F \wedge (X = \sigma)))$

        \State $\est \leftarrow \est + c $
        \EndFor

        \State\Return $\est$

    \end{algorithmic}
\end{algorithm}

\paragraph{Environment.}  All experiments were carried out on a cluster of nodes consisting of AMD EPYC 7713 CPUs running with 2x64 real cores. All tools were run in a single-threaded mode on a single core with a timeout of 10 hrs, i.e., 36000 seconds. A memory limit was set to 32 GB per core.
\paragraph{Parameters for Oracles and Implementation.} In the implementation, we utilized different state-of-the-art tools as counting and sampling oracles, including {\unisamp}~\cite{DM22} as an almost uniform sampling oracle,
and the latest version of {\apxmc}~\cite{YM23} as an approximate counting oracle.
{\sklmc} was tested with $\varepsilon = 0.8$ and $\delta = 0.4$. That gave the following values to error and tolerance parameters for model counting and sampling oracles.
The almost uniform sampling oracle {\unisamp} is run with $\eps=0.16$. The approximate model counting oracle {\apxmc} in line~\ref{algline:count} was run with $\epc = 4\sqrt{2} - 1$ and $\delta_c = \frac{0.4}{m \cdot s}$, where $s$  comes from the algorithm, based on input $(\varepsilon,\delta)$  and $m$ is number of output variables in the specification.
We carefully select error and tolerance values $\varepsilon_s, \varepsilon_c, \delta_c$ for counting and sampling oracles to ensure the validity of final bounds for {\sklmc} while also aiming for optimal performance of the counter based on these choices. The relationship between these values and the validity of bound of {\sklmc} is illustrated in the proof of \Cref{th:main}.

In our experiments, we sought to evaluate the run-time performance and approximation accuracy of {\sklmc}. Specifically, the following questions guided our investigation:

\begin{enumerate}[font=\bfseries RQ, leftmargin=\widthof{[RQQQ]}+\labelsep]
    \item How does {\sklmc} scale in terms of solving instances and the time taken in our benchmark set?

    \item What is the precision of the {\sklmc} approximation, and does it outperform its theoretical accuracy guarantees in practical scenarios?

\end{enumerate}

\paragraph{Benchmarks.} To evaluate the performance of {\sklmc}, we chose two sets of benchmarks.
\begin{enumerate}
    \item \emph{Efficiency benchmarks.}   609 instances from recent works on Boolean function synthesis~\cite{GRM20, ACJS17}, which includes different sources: the Prenex-2QBF track of QBF Evaluation 2017 and 2018, disjunctive~\cite{ACJS17}, arithmetic~\cite{TV17} and factorization~\cite{ACJS17}.

    \item \emph{Correctness Benchmarks.} The benchmarks described in the paragraph above are too hard for the baseline algorithm \cam{and {\qc}} to solve. As~\Cref{subsec:efficiency} reveals, the number of instances solved by the baseline is just eight out of the 609 instances.  Therefore, to check the correctness of {\sklmc} (RQ2), we used a set of 158 benchmarks from SyGuS instances~\cite{GRM21}. These benchmarks have very few input variables $(m \leq 8)$, and takes seconds for {\sklmc} to solve.

\end{enumerate}

\paragraph{Summary of Results.}  {\sklmc} achieves a huge improvement over {\qc} or {\base} by resolving {\tot} instances in a benchmark set consisting of 609, while {\base} only solved 8, \cam{and {\qc} solved none.} The accuracy of the approximate count is also noteworthy, with an average error of a count by {\sklmc} of only 21\%.

\subsection{Performance of {\sklmc}}\label{subsec:efficiency}
We evaluate the performance of {\sklmc} based on two metrics: the number of instances solved and the time taken to solve the instances.

\begin{table}[!tb]
    \centering
    \begin{tabular}{p{0.18\textwidth}r}
        \toprule
        Algorithm   & \# Instances solved \\ \midrule
        {\base} & 8                   \\
        {\sklmc}    & {\tot}                  \\ \bottomrule
    \end{tabular}
    \caption{Instances solved (out of 609). }
    \label{tab:result}
\end{table}

\begin{figure}[!t]
    \centering
    \includegraphics[width=0.7\linewidth]{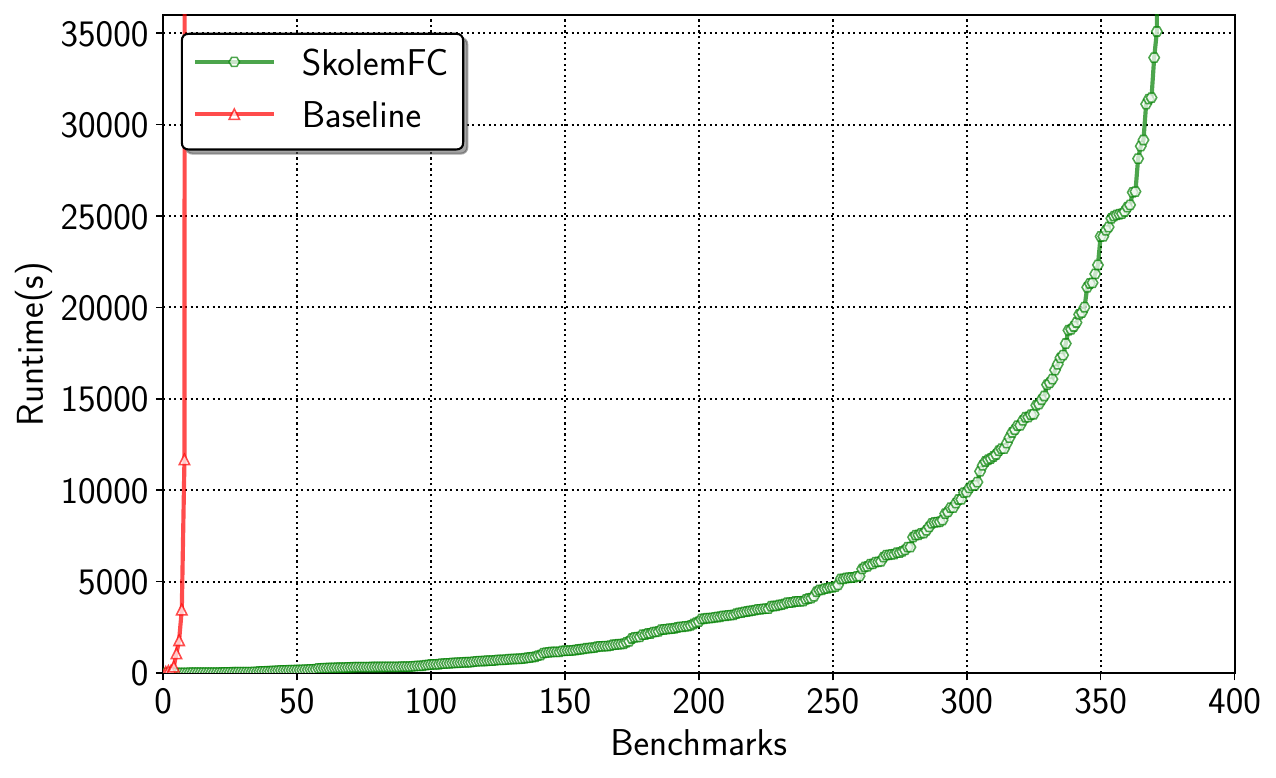}
    \caption{Runtime performance of {\sklmc} and {\base}. }
    \label{fig:cactus}
\end{figure}

\paragraph{Instances Solved.} In \Cref{tab:result}, we compare the number of benchmarks that can be solved by {\base} and {\sklmc}. First, it is evident that the {\base} only solved 8 out of the 609 benchmarks in the test suite, indicating its lack of scalability for practical use cases. Conversely, {\sklmc} solved {\tot} instances, demonstrating a substantial improvement compared to {\base}. \cam{{\qc} solves none of the instances from these benchmark sets. It is worth noting that on these benchmarks, there are many false QBFs, which are out of the scope of {\qc}.}

\paragraph{Solving Time Comparison.} A performance evaluation of {\base} and {\sklmc} is depicted in \Cref{fig:cactus}, which is a cactus plot comparing the solving time. The $x$-axis represents the number of instances, while the $y$-axis shows the time taken. A point $(i, j)$ in the plot represents that a solver solved $j$ benchmarks out of the 609 benchmarks in the test suite in less than or equal to $j$ seconds. The curves for {\base} and {\sklmc} indicate that for a few instances, {\base} was able to give a quick answer, while in the long run, {\sklmc} could solve many more instances given any fixed timeout.

\paragraph{Counter Call Comparison.} We analyze the algorithms' complexity in terms of counter calls, comparing {\base} and {\sklmc} across benchmarks in \Cref{fig:countercall}. The $x$ axis represents benchmarks, and the $y$ axis shows required counter calls, sorted by the increasing order of calls needed by {\base}. A red or green point $(i,j)$ signifies that {\base} or {\sklmc}, respectively, requires $j$ counting oracle calls for the $i^{th}$ instance.
{\base} requires up to a staggering $10^{230}$ counter calls for some instances, emphasizing the need for a scalable algorithm like {\sklmc}, which incurs significantly fewer counter calls.

\begin{figure}[tb]
    \centering
    \includegraphics[width=0.7\linewidth]{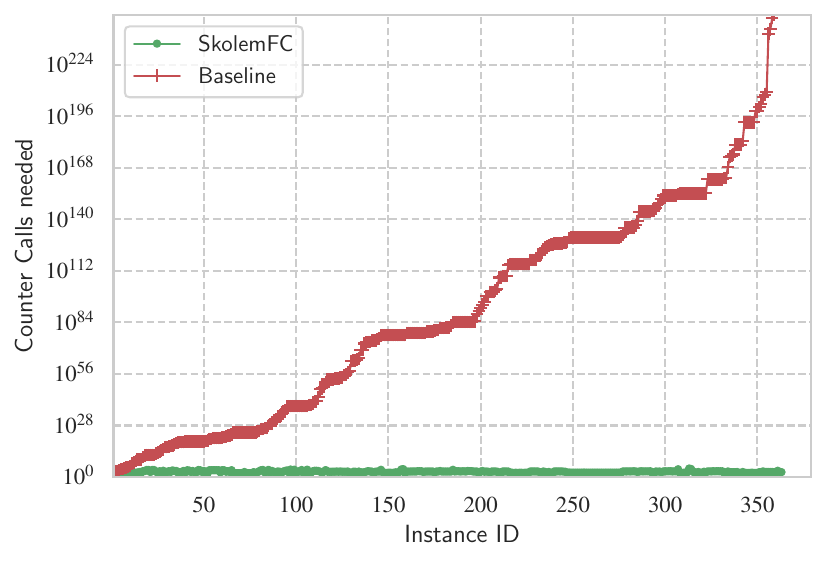}
    \caption{Counter calls needed by {\sklmc} and {\base} to solve the benchmarks.}
    \label{fig:countercall}
\end{figure}

We analyze the scalability of {\sklmc} by examining the correlation between average time per counter call and total counter calls, depicted in \Cref{fig:timeperiter}. The point $(i,j)$ means that if {\sklmc} needs $i$ counter calls, the average time per call is $j$ seconds. The figure showcases diverse scenarios: some with fewer iterations and longer durations per call, others with high counts and minimal time per call. %

\begin{figure}[!t]
    \centering
    \includegraphics[width=0.5\linewidth]{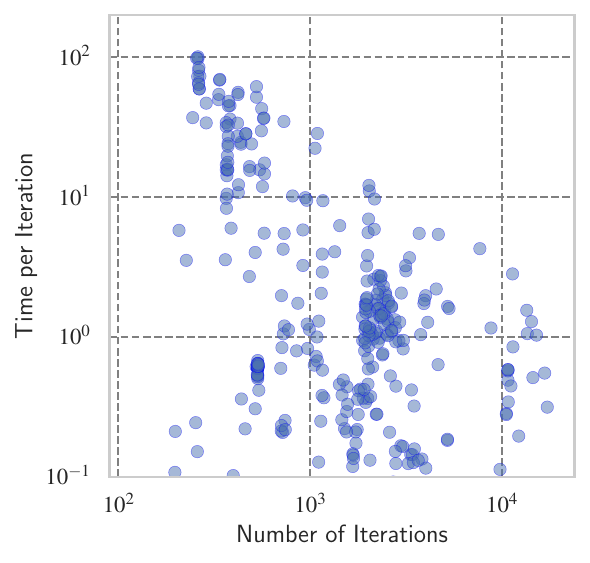}
    \caption{Relation between number of iterations needed by {\sklmc} and average time taken in each counter call.}
    \label{fig:timeperiter}
\end{figure}

\subsection{Quality of Approximation}

In the experiments, 158 accuracy benchmarks were measured using {\qc}, enabling comparison between {\qc} and {\sklmc} results, shown in \Cref{fig:count-comp}. The counts' close alignment and error reduction below theoretical guarantees were observed.
We quantify the {\sklmc} performance with error $e = \frac{|b - s|}{b}$, where $b$ is the count from {\qc} and $s$ from {\sklmc}. Analysis of all 158 cases found the average $e$ to be $0.005$, geometric mean $0.003$, and maximum $0.035$, contrasting sharply with a theoretical guarantee of $0.8$. This signifies {\sklmc} substantially outperforms its theoretical bounds.
Our findings underline {\sklmc}'s accuracy and potential as a dependable tool for various applications.

\begin{figure}[h!]
    \centering
    \includegraphics[width=0.7\linewidth]{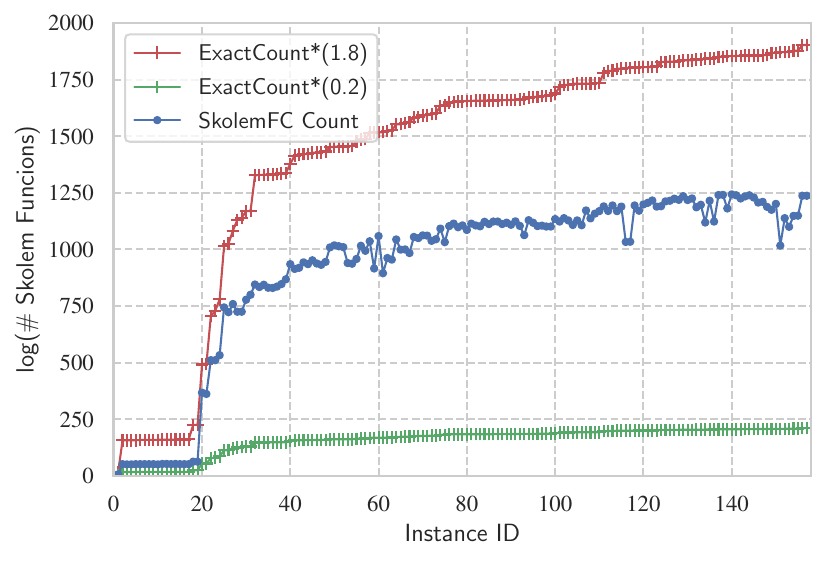}
    \caption{{\sklmc}'s estimate, vs. the theoretical bounds.}
    \label{fig:count-comp}
\end{figure}

\section{Conclusion} \label{sec:concl}

In conclusion, this paper presents the first scalable approximate Skolem function counter, {\sklmc}, which has been successfully tested on practical benchmarks and showed impressive performance. Our proposed method employs probabilistic techniques to provide theoretical guarantees for its results. The implementation leverages the progress made in the last two decades in the fields of constrained counting and sampling, and the practical results exceeded the theoretical guarantees. These findings open several directions for further investigation. One such area of potential extension is the application of the algorithm to other types of functions, such as counting uninterpreted functions in SMT with a more general syntax. This extension would enable the algorithm to handle a broader range of applications and provide even more accurate results. In summary, this research contributes significantly to the field of Skolem function counting and provides a foundation for further studies.

\section*{Acknowledgements}
We are thankful to Tim van Bremen, Priyanka Golia and Yash Pote for providing detailed feedback on the early drafts of the paper and grateful to the anonymous reviewers for their constructive comments to improve this paper. We thank  Martina Seidl, Andreas Plank and Sibylle Möhle  for pointing out an error in the first implementation of the code. This work was supported in part by National Research Foundation Singapore under its NRF Fellowship Programme [NRF-NRFFAI1-2019-0004], Ministry of Education Singapore Tier 2 Grant [MOE-T2EP20121-0011], Ministry of Education Singapore Tier 1 Grant [R-252-000-B59-114], and NSF awards IIS-1908287, IIS-1939677, and IIS-1942336. Part of the work was done during Arijit Shaw’s internship at the National University of Singapore. The computational work for this article was performed on resources of the National Supercomputing Centre, Singapore (\url{https://www.nscc.sg}).

\bibliographystyle{alpha}

\begin{thebibliography}{CMMV16}

\bibitem[AAC{\etalchar{+}}19]{AACKRS19}
S~Akshay, Jatin Arora, Supratik Chakraborty, S~Krishna, Divya Raghunathan, and
  Shetal Shah.
\newblock Knowledge compilation for {B}oolean functional synthesis.
\newblock In {\em Proc. of FMCAD}, 2019.

\bibitem[ABL13]{ABL13}
Carlos Ans{\'o}tegui, Maria~Luisa Bonet, and Jordi Levy.
\newblock {SAT-based MaxSAT algorithms}.
\newblock {\em Artificial Intelligence}, 2013.

\bibitem[ACG{\etalchar{+}}18]{ACGKS18}
S~Akshay, Supratik Chakraborty, Shubham Goel, Sumith Kulal, and Shetal Shah.
\newblock What’s hard about {B}oolean functional synthesis?
\newblock In {\em Proc. of CAV}, 2018.

\bibitem[ACJS17]{ACJS17}
S~Akshay, Supratik Chakraborty, Ajith~K John, and Shetal Shah.
\newblock Towards parallel boolean functional synthesis.
\newblock In {\em Proc. of TACAS}. Springer, 2017.

\bibitem[ADG16]{ADG16}
Aws Albarghouthi, Isil Dillig, and Arie Gurfinkel.
\newblock Maximal specification synthesis.
\newblock {\em ACM SIGPLAN Notices}, 2016.

\bibitem[BELM12]{BELM12}
Bernd Becker, R{\"u}diger Ehlers, Matthew Lewis, and Paolo Marin.
\newblock Allqbf solving by computational learning.
\newblock In {\em Proc. of ATVA}, 2012.

\bibitem[BJ11]{BJ11}
Valeriy Balabanov and Jie-Hong~R Jiang.
\newblock Resolution proofs and skolem functions in {QBF} evaluation and
  applications.
\newblock In {\em Proc. of CAV}, 2011.

\bibitem[BM15]{baudry2015multiple}
Benoit Baudry and Martin Monperrus.
\newblock The multiple facets of software diversity: Recent developments in
  year 2000 and beyond.
\newblock {\em ACM Computing Surveys (CSUR)}, 2015.

\bibitem[CCKM23]{CCKM23}
Diptarka Chakraborty, Sourav Chakraborty, Gunjan Kumar, and Kuldeep~S. Meel.
\newblock {Approximate Model Counting: Is SAT Oracle More Powerful than NP
  Oracle?}
\newblock In {\em Proc. of ICALP}, 2023.

\bibitem[CDM15]{CDM15}
Dmitry Chistikov, Rayna Dimitrova, and Rupak Majumdar.
\newblock {Approximate counting in SMT and value estimation for probabilistic
  programs}.
\newblock In {\em Proc. of TACAS}, 2015.

\bibitem[CFTV18]{CFTV18}
Supratik Chakraborty, Dror Fried, Lucas~M Tabajara, and Moshe~Y Vardi.
\newblock Functional synthesis via input-output separation.
\newblock In {\em Proc. of FMCAD}, 2018.

\bibitem[CMMV16]{CMMV16}
Supratik Chakraborty, Kuldeep Meel, Rakesh Mistry, and Moshe Vardi.
\newblock Approximate probabilistic inference via word-level counting.
\newblock In {\em Proc. of AAAI}, volume~30, 2016.

\bibitem[CMV13]{CMV13}
Supratik Chakraborty, Kuldeep~S Meel, and Moshe~Y Vardi.
\newblock A scalable approximate model counter.
\newblock In {\em Proc. of CP}, 2013.

\bibitem[CMV14]{CMV14}
Supratik Chakraborty, Kuldeep~S. Meel, and Moshe~Y. Vardi.
\newblock Balancing scalability and uniformity in sat-witness generator.
\newblock In {\em Proc. of DAC}, 6 2014.

\bibitem[CMV16]{CMV16}
Supratik Chakraborty, Kuldeep~S Meel, and Moshe~Y Vardi.
\newblock Algorithmic improvements in approximate counting for probabilistic
  inference: From linear to logarithmic sat calls.
\newblock In {\em Proc. of IJCAI}, 2016.

\bibitem[CMV21]{CMV21}
Supratik Chakraborty, Kuldeep~S Meel, and Moshe~Y Vardi.
\newblock Approximate model counting.
\newblock In {\em Handbook of Satisfiability}. 2021.

\bibitem[DHKV21]{DHKV21}
Arnaud Durand, Anselm Haak, Juha Kontinen, and Heribert Vollmer.
\newblock Descriptive complexity of \#{P} functions: A new perspective.
\newblock {\em Journal of Computer and System Sciences}, 2021.

\bibitem[DKLR95]{DKLR95}
P~Dagum, R~Karp, M~Luby, and S~Ross.
\newblock An optimal algorithm for monte carlo estimation.
\newblock In {\em Proc. of FOCS}, 1995.

\bibitem[DL97]{DL97}
Paul Dagum and Michael Luby.
\newblock An optimal approximation algorithm for bayesian inference.
\newblock {\em Artificial Intelligence}, 1997.

\bibitem[DM22]{DM22}
Remi Delannoy and Kuldeep~S Meel.
\newblock On almost-uniform generation of sat solutions: The power of 3-wise
  independent hashing.
\newblock In {\em Proc. of LICS}, 2022.

\bibitem[DPV20]{DPV20}
Jeffrey~M. Dudek, Vu~H.N. Phan, and Moshe~Y. Vardi.
\newblock {ADDMC: Weighted model counting with algebraic decision diagrams}.
\newblock In {\em Proc. of AAAI}, 2020.

\bibitem[FHI{\etalchar{+}}21]{FHIJ21}
Nils Froleyks, Marijn Heule, Markus Iser, Matti J{\"a}rvisalo, and Martin Suda.
\newblock Sat competition 2020.
\newblock {\em Artificial Intelligence}, 2021.

\bibitem[GRM20]{GRM20}
Priyanka Golia, Subhajit Roy, and Kuldeep~S Meel.
\newblock Manthan: A data-driven approach for boolean function synthesis.
\newblock In {\em Proc. of CAV}, 2020.

\bibitem[GRM21]{GRM21}
Priyanka Golia, Subhajit Roy, and Kuldeep~S. Meel.
\newblock Program synthesis as dependency quantified formula modulo theory.
\newblock In {\em Proc. of IJCAI}, 2021.

\bibitem[GSRM21]{GSRM21}
Priyanka Golia, Friedrich Slivovsky, Subhajit Roy, and Kuldeep~S Meel.
\newblock Engineering an efficient boolean functional synthesis engine.
\newblock In {\em Proc. of ICCAD}, 2021.

\bibitem[GSS21]{GSS21}
Carla~P Gomes, Ashish Sabharwal, and Bart Selman.
\newblock Model counting.
\newblock In {\em Handbook of satisfiability}. IOS press, 2021.

\bibitem[HV19]{HV19}
Anselm Haak and Heribert Vollmer.
\newblock A model-theoretic characterization of constant-depth arithmetic
  circuits.
\newblock {\em Annals of Pure and Applied Logic}, 2019.

\bibitem[Jia09]{J09}
Jie-Hong~R Jiang.
\newblock Quantifier elimination via functional composition.
\newblock In {\em Proc. of CAV}, 2009.

\bibitem[JLH09]{JLH09}
Jie{-}Hong~Roland Jiang, Hsuan{-}Po Lin, and Wei{-}Lun Hung.
\newblock Interpolating functions from large boolean relations.
\newblock In {\em Proc. of {ICCAD}}, 2009.

\bibitem[JSC{\etalchar{+}}15]{JSCTA15}
Ajith~K John, Shetal Shah, Supratik Chakraborty, Ashutosh Trivedi, and
  S~Akshay.
\newblock Skolem functions for factored formulas.
\newblock In {\em Proc. of FMCAD}, 2015.

\bibitem[KS16]{KS16}
Daniel Kroening and Ofer Strichman.
\newblock {\em Decision procedures}.
\newblock Springer, 2016.

\bibitem[LM21]{LM21}
Chu~Min Li and Felip Manya.
\newblock Maxsat, hard and soft constraints.
\newblock In {\em Handbook of satisfiability}. 2021.

\bibitem[PFMD21]{PFMD21}
Sumanth Prabhu, Grigory Fedyukovich, Kumar Madhukar, and Deepak D'Souza.
\newblock Specification synthesis with constrained horn clauses.
\newblock In {\em Proc. of PLDI}, 2021.

\bibitem[PMS23]{PMS23}
Andreas Plank, Sibylle M{\"o}hle, and Martina Seidl.
\newblock Enumerative level-2 solution counting for quantified boolean formulas
  (short paper).
\newblock In {\em Proc. of CP}, 2023.

\bibitem[Rab19]{R19}
Markus~N Rabe.
\newblock Incremental determinization for quantifier elimination and functional
  synthesis.
\newblock In {\em Proc. of CAV}, 2019.

\bibitem[RS16]{RS16}
Markus~N. Rabe and Sanjit~A. Seshia.
\newblock Incremental determinization.
\newblock In {\em Proc. of {SAT}}, 2016.

\bibitem[RTRS18]{RTRS18}
Markus~N Rabe, Leander Tentrup, Cameron Rasmussen, and Sanjit~A Seshia.
\newblock Understanding and extending incremental determinization for {2QBF}.
\newblock In {\em Proc. of CAV}, 2018.

\bibitem[SM19]{SM19}
Mate Soos and Kuldeep~S Meel.
\newblock {BIRD: engineering an efficient CNF-XOR SAT solver and its
  applications to approximate model counting}.
\newblock In {\em Proc. of AAAI}, 2019.

\bibitem[SMKS22]{SMKS22}
Ankit Shukla, Sibylle M{\"o}hle, Manuel Kauers, and Martina Seidl.
\newblock Outercount: A first-level solution-counter for quantified boolean
  formulas.
\newblock In {\em Proc. of CICM}, 2022.

\bibitem[SNC09]{SNC09}
Mate Soos, Karsten Nohl, and Claude Castelluccia.
\newblock Extending sat solvers to cryptographic problems.
\newblock In {\em Proc. of SAT}, 2009.

\bibitem[SRSM19]{SRMM19}
Shubham Sharma, Subhajit Roy, Mate Soos, and Kuldeep~S Meel.
\newblock {GANAK: A Scalable Probabilistic Exact Model Counter.}
\newblock In {\em Proc. of IJCAI}, 2019.

\bibitem[Sto83]{S83}
Larry Stockmeyer.
\newblock The complexity of approximate counting.
\newblock In {\em Proc. of ACM symposium on Theory of computing}, 1983.

\bibitem[Thu06]{T06}
Marc Thurley.
\newblock {sharpSAT--counting models with advanced component caching and
  implicit BCP}.
\newblock In {\em Proc. of SAT}, 2006.

\bibitem[TV17]{TV17}
Lucas~M Tabajara and Moshe~Y Vardi.
\newblock Factored boolean functional synthesis.
\newblock In {\em Proc. of FMCAD}. IEEE, 2017.

\bibitem[YM23]{YM23}
Jiong Yang and Kuldeep~S Meel.
\newblock {Rounding Meets Approximate Model Counting}.
\newblock In {\em Proc. of CAV}, 2023.

\end{thebibliography}
\newcommand{\etalchar}[1]{$^{#1}$}

\end{document}